\documentclass{article}

\usepackage[margin=1in]{geometry}
\usepackage{microtype}

\usepackage[hidelinks]{hyperref}

\title{Fractal Intersections and Products via Algorithmic Dimension}

\author{Neil Lutz}

\date{}

\usepackage{amsmath}
\usepackage{amssymb}
\usepackage{mdframed}
\usepackage{amsthm}
\usepackage{mathtools}
\usepackage{placeins}
\usepackage{tikz}

\usetikzlibrary{lindenmayersystems}

\newtheorem{thm}{Theorem}

\newtheorem{lem}[thm]{Lemma}

\newtheorem{cor}[thm]{Corollary}

\theoremstyle{remark}

\theoremstyle{definition}
\newtheorem{definition}{Definition}

\DeclareMathOperator{\Dim}{Dim}
\DeclareMathOperator{\diam}{diam}

\newcommand{\R}{\mathbb{R}}

\newcommand{\N}{\mathbb{N}}
\newcommand{\Q}{\mathbb{Q}}

\newcommand{\ve}{\varepsilon}

\newcommand{\dimH}{\mathrm{dim}_\mathrm{H}}
\newcommand{\dimP}{\mathrm{dim}_\mathrm{P}}

\pgfdeclarelindenmayersystem{Koch}{
	\symbol{F}{\pgflsystemdrawforward}
	\rule{F -> F-F++F-F}
}
\def\ordervar{5}

\begin{document}
	\maketitle

\begin{abstract} 
	Algorithmic fractal dimensions quantify the algorithmic information density of individual points and may be defined in terms of Kolmogorov complexity. This work uses these dimensions to bound the classical Hausdorff and packing dimensions of intersections and Cartesian products of fractals in Euclidean spaces. This approach shows that two prominent, fundamental results about the dimension of Borel or analytic sets also hold for arbitrary sets.
\end{abstract}

	\section{Introduction}
		Classical fractal dimensions, among which \emph{Hausdorff dimension}~\cite{Haus19} and \emph{packing dimension}~\cite{Tric82} are the most important, refine notions of measure to quantitatively classify sets of measure 0. In 2000, J. H. Lutz~\cite{Lutz03b} showed that Hausdorff dimension can be simply characterized using betting strategies called \emph{gales}, and that this characterization can be effectivized in order to quantitatively classify non-random infinite data objects. This approach yielded \emph{algorithmic fractal dimensions}, which have been applied to multiple areas of computer science and have proved especially useful in algorithmic information theory~\cite{Mayo08}.
		
		The connection between algorithmic and classical dimensions has more recently been exploited in the other direction, i.e., to apply algorithmic information theoretic methods and intuition to classical fractal geometry (e.g.,~\cite{Reim08,BeReSl18}). A \emph{point-to-set principle} of J. H. Lutz and N. Lutz~\cite{LutLut18}, stated here as Theorem~\ref{thm:p2s}, characterizes the classical Hausdorff dimension of any set in $\R^n$ in terms of the algorithmic dimensions of its individual points.
		
		In the same work, J. H. Lutz and N. Lutz showed that the point-to-set principle gives rise to a new, pointwise technique for dimensional lower bounds, and, as a proof of concept, used this technique to give an algorithmic information theoretic proof of Davies's 1971 theorem~\cite{Davi71} stating that every Kakeya set in $\R^2$ has Hausdorff dimension 2. This bounding technique has since been used by N. Lutz and Stull~\cite{LutStu20} to make new progress on a problem in classical fractal geometry by deriving an improved lower bound on the Hausdorff dimension of generalized Furstenberg sets, as defined by Molter and Rela~\cite{MolRel12}.
		
		The same algorithmic dimensional technique is applied here to intersections and products of fractals, yielding two main results. First, we extend the following intersection formula, previously shown to hold when $E$ and $F$ are Borel sets~\cite{Falc14}, to arbitrary sets $E$ and $F$.\footnote{This result is closely related to the Marstrand's slicing theorem, as stated in the excellent recent book by Bishop and Peres~\cite{BisPer16}. The proof given there assumes that a set is Borel, but this assumption was inadvertently omitted from the theorem statement~\cite{Bishop17}.} This formula is illustrated in Figure~\ref{fig:int}.
		\begin{thm}\label{thm:hausint}
			For all $E,F\subseteq\R^n$, and for (Lebesgue) almost all $z\in\R^n$,
			\begin{equation*}
			\dimH(E\cap(F+z))\leq\max\{0,\dimH(E\times F)-n\}\,,
			\end{equation*}
			where $F+z=\{x+z:x\in F\}$. 
		\end{thm}
		\begin{figure}[ht]
			\centering
			\begin{tikzpicture}[scale=1.25]
			\draw [gray,very thin,shift={(7.5,0)}] l-system
			[l-system={Koch, axiom=F++F++F,step=3cm/3^\ordervar, order=\ordervar, angle=60}];
			\draw [black,very thin,shift={(5.5,0)}] (0,0) l-system
			[l-system={Koch, axiom=F++F++F,step=3cm/3^\ordervar, order=\ordervar, angle=60}];
			\draw [gray,very thin,shift={(1.43,0.13)}] (0,0) l-system
			[l-system={Koch, axiom=F++F++F,step=3cm/3^\ordervar, order=\ordervar, angle=60}];
			\draw [black,very thin,shift={(0,0)}] (0,0) l-system
			[l-system={Koch, axiom=F++F++F,step=3cm/3^\ordervar, order=\ordervar, angle=60}];
			\end{tikzpicture}
			\caption{Let $E$ and $F$ both be Koch snowflakes, which have Hausdorff dimension $\log_34\approx1.26$. \emph{Left:} Theorem~\ref{thm:hausint} states that, for almost all translation parameters $z\in\R^2$, the Hausdorff dimension of the intersection $E\cap(F+z)$ is at most $2\log_34-2\approx0.52$. \emph{Right:} For a measure zero set of translations, the intersection may have Hausdorff dimension as large as that of the original sets. Note that Koch curves are Borel sets, so the new generality introduced by Theorem~\ref{thm:hausint} is not required for this example.}\label{fig:int}
		\end{figure}

		Second, we prove that the following characterization of packing dimension, previously known to hold when $E$ is an analytic set, holds for arbitrary $E$.

		\begin{thm}\label{thm:packprod}
			For every set $E\subseteq\R^n$,
			\[\dimP(E)=\sup_{F\subseteq\R^n}\big(\dimH(E\times F)-\dimH(F)\big)\,.\]
		\end{thm}
	
		Our approach also yields a simplified proof of Marstrand's~\cite{Mars54} product formula for general sets:
		\begin{thm}[\cite{Mars54}]\label{thm:mars}
			For all $E\subseteq\R^m$ and $F\subseteq\R^n$,
			\[\dimH(E)+\dimH(F)\leq\dimH(E\times F)\,.\]
		\end{thm}
		There are previously known analogues to Theorems~\ref{thm:hausint} and~\ref{thm:mars} for packing dimension~\cite{Tric82,Falc94}, and we use symmetric arguments to derive these as well. These results are included here to showcase the versatility of this technique and its ability to capture the exact duality between Hausdorff and packing dimensions.
		
	\section{Classical Fractal Dimensions}\label{sec:classical}
	
		We begin by stating classical, measure-theoretic definitions of the two most well-studied notions of fractal dimension, Hausdorff dimension and packing dimension. These definitions are included here for completeness but are not used directly in the remainder of this work; we will instead apply equivalent characterizations in terms of algorithmic information, as described in Section~\ref{sec:dim}.

		\begin{definition}[\cite{Haus19}]
			For $E\subseteq\R^n$, let $\mathcal{U}_\delta(E)$ be the collection of all countable covers of $E$ by sets of positive diameter at most $\delta$.
			For all $s\geq 0$, let
			\[H_\delta^s(E)=\inf\left\{\sum_{i\in\N}\diam(U_i)^s\,:\,\{U_i\}_{i\in\N}\in\mathcal{U}_\delta(E)\right\}\,.\]
			The \emph{$s$-dimensional Hausdorff (outer) measure of $E$} is
			\[H^s(E)=\lim_{\delta\to 0^+}H_\delta^s(E)\,,\]
			and the \emph{Hausdorff dimension of $E$} is
			\[\dimH(E)=\inf\left\{s>0:H^s(E)=0\right\}\,.\]
		\end{definition}
	
		Three desirable properties have made $\dimH$ the most standard notion of fractal dimension since it was introduced by Hausdorff in 1919~\cite{Haus19}. First, it is defined on every set in $\R^n$. Second, it is \emph{monotone}: if $E\subseteq F$, then $\dimH(E)\leq\dimH(F)$. Third, it is \emph{countably stable}: if $E=\bigcup_{i\in\N}E_i$, then $\dimH(E)=\sup_{i\in\N}\dimH(E_i)$. These three properties also hold for packing dimension, which was defined much later, independently by Tricot~\cite{Tric82} and by Sullivan~\cite{Sull84}.

		\begin{definition}[\cite{Tric82,Sull84}]
			For all $x\in\R^n$ and $\rho>0$, let $B_\rho(x)$ denote the open ball of radius $\rho$ and center $x$. For all $E\subseteq\R^n$, let $\mathcal{V}_\delta(E)$ be the class of all countable collections of pairwise disjoint open balls with centers in $E$ and diameters at most $\delta$. That is, for every $i\in\N$, we have $V_i=B_{\rho_i}(x_i)$ for some $x_i\in E$ and $\rho_i\in[0,\delta/2]$, and for every $j\neq i$, $V_i\cap V_j=\emptyset$.
			For all $s\geq0$, define
			\[P_\delta^s(E)=\sup\left\{\sum_{i\in\N}\diam(V_i)^s\,:\,\{V_i\}_{i\in\N}\in\mathcal{V}_\delta(E)\right\}\,,\]
			and let
			\[P_0^s(E)=\lim_{\delta\to0^+}P^s_\delta(E)\,.\]
			The \emph{$s$-dimensional packing (outer) measure of $E$} is
			\[P^s(E)=\inf\bigg\{\sum_{i\in\N} P_0^s(E_i)\,:\,E\subseteq\bigcup_{i\in\N}E_i\bigg\}\,,\]
			and the \emph{packing dimension of $E$} is
			\[\dimP(E)=\inf\left\{s:P^s(E)=0\right\}\,.\]
		\end{definition}
		Notice that defining packing dimension in this way requires an extra step of optimization compared to Hausdorff dimension. More properties and details about classical fractal dimensions may be found in standard references such as~\cite{Matt95,Falc14,SteSha05}.
		
	\section{Algorithmic Fractal Dimensions}\label{sec:dim}
		This section defines algorithmic fractal dimensions in terms of algorithmic information, i.e., Kolmogorov complexity. We also discuss some properties of these dimensions, including their relationships to classical Hausdorff and packing dimensions.

		\subsection{Kolmogorov Complexity}
				
		Kolmogorov complexity quantifies the \emph{incompressibility} of finite data objects.

		\begin{definition}[see~\cite{LiVit08}]
			The \emph{(prefix-free) conditional Kolmogorov complexity} of a string $\sigma\in\{0,1\}^*$ given another string $\tau\in\{0,1\}^*$ is the length of the shortest binary program that outputs $\sigma$ when given $\tau$ as an input. Formally, it is
			\[K(\sigma\mid \tau)=\min_{\pi\in\{0,1\}^*}\{|\pi|\,:\,U(\pi,\tau)=\sigma\}\,,\]
			where $U$ is a fixed universal Turing machine that is prefix-free in its first input and $|\pi|$ denotes the length of the string $\pi$.
		
			The \emph{Kolmogorov complexity of $\sigma\in\{0,1\}^*$} is conditional Kolmogorov complexity of $\sigma\in\{0,1\}^*$ given $\lambda$, the empty string:
			\[K(\sigma)=K(\sigma\mid\lambda)\,.\]
			This quantity is also called the \emph{algorithmic information content} of $\sigma$.
		\end{definition}

		The \emph{symmetry of algorithmic information} (Levin and Kolmogorov, see~\cite{ZvoLev70}) is an essential property of conditional Kolomogorov complexity stating that for all $\sigma,\tau\in\{0,1\}^*$,
		\begin{equation}\label{eq:symmetry}
		K(\sigma\tau) = K(\sigma\mid\tau) + K(\tau) + O(\log |\sigma\tau|)\,.
		\end{equation}

		The definitions in this subsection are readily extended to discrete domains other than binary strings. For the purposes of this work, the complexity of rational points is most relevant. Hence, fix some standard binary encoding for $n$-tuples of rationals.

		\subsection{Algorithmic Dimensions of a Point}
		Using approximation by rationals, Kolmogorov complexity may be further extended to Euclidean spaces~\cite{LutMay08}. For every $E\subseteq\R^n$, define
		\[K(E)=\min\{K(p):p\in E\cap\Q^n\}\,,\]
		where the minimum is understood to be infinite if $E\cap\Q^n$ is empty. This is the length of the shortest program that outputs some rational point in $E$.
		The \emph{Kolmogorov complexity} of $x\in\R^n$ at \emph{precision} $r\in\N$ is given by
		\[K_r(x)=K(B_{2^{-r}}(x))\,,\]
		the length of the shortest program that outputs any \emph{precision-$r$} rational approximation of $x$. $K_r(x)$ may also be described as the \emph{algorithmic information content} of $x$ at precision $r$, and similarly, $K_r(x)/r$ is the \emph{algorithmic information density} of $x$ at precision $r$. This ratio does not necessarily converge as $r\to\infty$, but it does have limiting bounds in $[0,n]$. These limits are used to define algorithmic dimensions.
		\begin{definition}[\cite{Lutz03b,Mayo02,AHLM07,LutMay08}]\label{def:edims}
			Let $x\in\R^n$.
			\begin{enumerate}
			\item The \emph{(algorithmic) dimension} of $x$ is
			\[\dim({x})=\liminf_{r\to\infty}\frac{K_r({x})}{r}\,.\]
			\item The \emph{strong (algorithmic) dimension} of $x$ is \[\Dim({x})=\limsup_{r\to\infty}\frac{K_r({x})}{r}\,.\]
			\end{enumerate}
		\end{definition}
	
		These dimensions were originally defined by J. H. Lutz~\cite{Lutz03b} and Athreya, Hitchcock, J. H. Lutz, and Mayordomo~\cite{AHLM07}, respectively. The original definitions were in Cantor space and used \emph{gales}, which are betting strategies that generalize martingales, emphasizing the unpredictability of a sequence instead of its incompressibility. The Kolmogorov complexity characterizations and translation to Euclidean spaces are due to Mayordomo~\cite{Mayo02} and J. H. Lutz and Mayordomo~\cite{LutMay08}. Relationships between Hausdorff dimension and Kolmogorov complexity were also studied earlier by Ryabko~\cite{Ryab84,Ryabko86,Ryabko93,Ryabko94}, Staiger~\cite{Staiger89,Staiger98}, and Cai and Hartmanis~\cite{CaiHar94}; see Section~6 of~\cite{Lutz03b} for a detailed discussion of this history.

		We will use the fact that these dimensions are preserved by sufficiently well-behaved functions, namely bi-Lipschitz computable functions.
		\begin{lem} [\cite{Reim04,CasLut15}]\label{lem:pres}
			If $f:\R^m\to\R^n$ is computable and bi-Lipschitz, then $\dim(x)=\dim(f(x))$ and $\Dim(x)=\Dim(f(x))$ for all $x\in\R^m$.
		\end{lem}

		It will sometimes be convenient to identify a Euclidean point $x\in\R^n$ with an infinite binary sequence that interleaves the binary expansions of $x$'s coordinates to the right of the binary point. That is, if $x=(x^1,x^2,\ldots,x^n)$, where each $x^i=x^i_{-k_i}x^i_{1-k_i}\ldots x^i_0\,.\,x^i_1x^i_2\ldots$, for some $k_i\in\N$, then we associate $x$ with the binary sequence
		\[\dot{x}=x^1_1x^2_1\ldots x^{n-1}_1x^n_1x^1_2x^2_2\ldots x^{n-1}_2x^n_2\ldots\,.\]
		Discarding the bits to the left of the binary point is convenient and will affect the complexity of prefixes of $\dot{x}$ by only a constant amount. Given an infinite binary sequence $y=y_1y_2y_3\ldots$ and $a,b\in\N$, we let $y[a:b]$ denote $y_{a+1}y_{a+2}\ldots y_{b-1}y_{b}$. If $a\geq b$, then $y[a:b]$ is the empty string $\lambda$. Thus, for any point $x\in[0,1]^n$ and precision parameter $r\in\N$, the string $\dot{x}[0:nr]$ specifies a dyadic point that is within distance $2^{-r}\sqrt{n}$ of $x$. Naturally,
		\begin{equation}\label{eq:convert}
			K(\dot{x}[0:nr])=K_r(x)+O(\log r)\,.
		\end{equation}
		(Throughout this work, the dimensions $m$ and $n$ of Euclidean spaces are treated as constants in asymptotic notation.) See~\cite{LutStu20} for a detailed analysis of the relationship between these two quantities.

		\subsection{Conditional and Relative Dimensions}
		By analogy with Definition~\ref{def:edims}, J. H. Lutz and N. Lutz~\cite{LutLut18} defined \emph{conditional dimensions}. For every $E\subseteq\R^m$ and $F\subseteq\R^n$, define
		\[K(E\mid F)=\max_{q\in F\cap\Q^n}\min_{p\in E\cap\Q^m} K(p\mid q)\,,\]
		where this quantity is understood to be infinite if $E\cap\Q^m$ is empty and $K(E)$ if $F\cap\Q^n$ is empty. This is the length of the shortest program that outputs some rational point in $E$ when given the ``least helpful'' rational point from $F$ as an input. The \emph{conditional Kolmogorov complexity} of $x\in\R^m$ given $y\in\R^n$ at precision $r\in \N$ is given by
		\[K_r(x\mid y)=K(B_{2^{-r}}(x)\mid B_{2^{-r}}(y))\,,\]
		the length of the shortest program that outputs any precision-$r$ rational approximation of $x$, given an adversarially selected precision-$r$ rational approximation of $y$. These quantities can also be put in terms of binary expansions with only logarithmic additive error~\cite{LutStu20}: For all $x\in\R^m$, $y\in\R^n$, and $r\in\N$,
		\begin{equation}\label{eq:condconvert}
			K(\dot{x}[0:mr]\mid\dot{y}[0:nr])=K_r(x\mid y)+O(\log r)\,.
		\end{equation}

		\begin{definition}[\cite{LutLut18}]\label{def:cdims}
			Let $x\in\R^m$ and $y\in\R^n$. The \emph{lower} and \emph{upper conditional dimensions} of $x$ given $y$ are
			\[\dim(x\mid y)=\liminf_{r\to\infty}\frac{K_r(x\mid y)}{r}\quad\text{and}\quad\Dim(x\mid y)=\limsup_{r\to\infty}\frac{K_r(x\mid y)}{r}\,,\]
			respectively.
		\end{definition}
	
		Conditional dimensions obey the following \emph{chain rule for dimensions}, which plays a key role in the proofs of our main results.
				
		\begin{thm}[\cite{LutLut18}]\label{thm:chainrule}
			For all $x\in\R^m$ and $y\in\R^n$,
			\begin{align*}
			\dim(x\mid y)+\dim(y)&\leq\dim(x, y)\\
			&\leq\Dim(x\mid y)+\dim(y)\\
			&\leq\Dim(x,y)\\
			&\leq\Dim(x\mid y)+\Dim(y)\,.
			\end{align*}
		\end{thm}

		In addition to conditional dimensions, which are based on bounded-precision access to the given point $y$, we consider \emph{relativized} dimensions, which allow for arbitrary access to a given point or set. By making the fixed universal machine $U$ an oracle machine, the algorithmic information quantities above may be defined relative to any oracle $A\subseteq\N$. The definitions of $K^A(\sigma)$, $K^A(\sigma\mid \tau)$, $K^A_r(x)$, $K^A_{r}(x\mid y)$, $\dim^A({x})$, $\Dim^A({x})$, $\dim^A(x\mid y)$, and $\Dim^A(x\mid y)$ all exactly mirror their unrelativized versions, except that $U$ is permitted to query membership in $A$ as a computational step. For $y\in\R^n$, we let $A_y\subseteq\N$ be the set defined by $i\in A_y\iff \dot{y}_i=1$, and we write $\dim^y(x)$ and $\Dim^y(x)$ as shorthand for $\dim^{A_y}(x)$ and $\Dim^{A_y}(x)$.

		Relativization and conditioning never increase the dimension of a point; the length of an instruction to ignore the additional information is asymptotically negligible. The following lemma formalizes the related intuition that, when approximating a point $x$, arbitrary access to another point $y$ is asymptotically at least as helpful as bounded-precision access to $y$.

		\begin{lem}[\cite{LutLut18}]\label{lem:relcond}
			If $x\in\R^m$, $y\in\R^n$, and $r\in\N$, then $K_r^y(x)\leq K_r(x\mid y)+O(\log n)$, and therefore $\dim^y(x)\leq\dim(x\mid y)$ and $\Dim^y(x)\leq\Dim(x\mid y)$.
		\end{lem}

		\subsection{Point-to-Set Principle}
		Algorithmic fractal dimensions were conceived as effective versions of classical Hausdorff dimension and packing dimension~\cite{Lutz03b,AHLM07}. J. H. Lutz and N. Lutz~\cite{LutLut18} established the following \emph{point-to-set principle}, which uses relativization to precisely characterize the relationship of these dimensions to their non-algorithmic precursors.
		\begin{thm}[\cite{LutLut18}]\label{thm:p2s}
			 For every $E\subseteq\R^n$, the Hausdorff dimension and packing dimension of $E$ are
			\begin{enumerate}
				\item $\displaystyle \dimH(E)=\adjustlimits\min_{A\subseteq\N}\sup_{{x}\in E}\,\dim^A({x})\,,$
				\item $\displaystyle\dimP(E)=\adjustlimits\min_{A\subseteq\N}\sup_{{x}\in E}\,\Dim^A({x})\,.$
			\end{enumerate}
		\end{thm}

		We call $A\subseteq\N$ a \emph{Hausdorff oracle} for a set $E\subseteq\R^n$ if $\dimH(E)=\sup_{x\in E}\Dim^A(x)$ and a \emph{packing oracle} for $E$ if $\dimP(E)=\sup_{x\in E}\dim^A(x)$. Notice that because relativization cannot increase dimension, Theorem~\ref{thm:p2s} implies that pairing a Hausdorff oracle for $E$ with any other oracle results in another Hausdorff oracle for $E$, and similarly for packing oracles.
		
		This theorem allows us to prove lower bounds on classical dimensions in a pointwise manner. In a typical proof using this technique, we apply Theorem~\ref{thm:p2s} twice, in two different ways. For the sake of clarity about how it is being applied each time, we enumerate its immediate consequences in the following corollary.

		\begin{cor}\label{cor:p2suses}
			Let $E\subseteq\R^n$.
			\begin{enumerate}
				\item There is an oracle $A\subseteq \N$ such that $\dimH(E)=\sup_{x\in E}\dim^A(x)$.
				\item For every $A\subseteq \N$ and every $\varepsilon>0$, there is a point $x\in E$ such that $\dim^A(x)>\dimH(E)-\ve$.
				\item There is an oracle $A\subseteq \N$ such that $\dimP(E)=\sup_{x\in E}\Dim^A(x)$.
				\item For every $A\subseteq \N$ and every $\varepsilon>0$, there is a point $x\in E$ such that $\Dim^A(x)>\dimP(E)-\ve$.
			\end{enumerate}
		\end{cor}

		Unlike earlier applications of this bounding technique~\cite{LutLut18,LutStu20}, the proofs in Sections~\ref{sec:int} and~\ref{sec:prod} do not directly invoke Kolmogorov complexity; the only tools needed are Lemma~\ref{lem:pres}, Theorem~\ref{thm:chainrule}, Lemma~\ref{lem:relcond}, and Theorem~\ref{thm:p2s} (in the form of Corollary~\ref{cor:p2suses}). The proof of our product characterization of packing dimension in Section~\ref{sec:packprod} employs both the point-to-set principle and a construction that explicitly uses Kolmogorov complexity.
		
	\section{Intersections of Fractals}\label{sec:int}
		In this section we prove Theorem~\ref{thm:hausint}. We then use a symmetric argument to give a new proof of the corresponding statement for packing dimension, which was originally proved by Falconer~\cite{Falc94}.
		The case of Theorem~\ref{thm:hausint} in which $E,F\subseteq\R^n$ are Borel sets was also proved by Falconer~\cite{Falc14}. The special case where $E\subseteq\R^2$ is a Borel set and $F$ is a vertical line is Marstrand's slicing theorem~\cite{Mars54,BisPer16}. Closely related results, which also place restrictions on $E$ and $F$,  were proved by Mattila~\cite{Matt84,Matt85} and Kahane~\cite{Kaha86}.
		{\renewcommand{\thethm}{\ref{thm:hausint}}
		\begin{thm}
			For all $E,F\subseteq\R^n$, and for (Lebesgue) almost all $z\in\R^n$,
			\begin{equation}\label{eq:hausint1}
			\dimH(E\cap(F+z))\leq\max\{0,\dimH(E\times F)-n\}\,,
			\end{equation}
			where $F+z=\{x+z:x\in F\}$. 
		\end{thm}
		\addtocounter{thm}{-1}}
		\begin{proof}
			Let $E,F\subseteq\R^n$ and $z\in\R^n$. If $E\cap(F+z)=\emptyset$, then inequality~\eqref{eq:hausint1} holds trivially, so assume that the intersection is nonempty.
			Applying Corollary~\ref{cor:p2suses}(1), let $A\subseteq\N$ be a Hausdorff oracle for $E\times F$, i.e.,
			\begin{equation}\label{eq:hausint2}
			\dimH(E\times F)=\sup_{(x,y)\in E\times F}\dim^A(x,y)\,.
			\end{equation}
			Fix $\varepsilon>0$. Corollary~\ref{cor:p2suses}(2), applied to $ E\cap(F+z)$, gives a point $x\in E\cap(F+z)$ such that
			\begin{equation}\label{eq:hausint3}
			\dim^{A,z}(x)\geq\dimH(E\cap(F+z))-\ve\,.
			\end{equation}
			Since $(x,x-z)\in E\times F$, we have
			\begin{align*}
				\dimH(E\times F)&\geq\dim^A(x,x-z)\tag{by equation~\eqref{eq:hausint2}}\\
				&=\dim^A(x,z)\tag{by Lemma~\ref{lem:pres}}\\
				&\geq\dim^A(z)+\dim^{A}(x\mid z)\tag{by Theorem~\ref{thm:chainrule} relative to $A$}\\
				&\geq\dim^A(z)+\dim^{A,z}(x)\tag{by Lemma~\ref{lem:relcond} relative to $A$}\\
				&\geq\dim^A(z)+\dimH(E\cap(F+z))-\ve\,.\tag{by inequality~\eqref{eq:hausint3}}
			\end{align*}
			Letting $\ve\to 0$, we have
			\[\dimH(E\cap(F+z))\leq\dimH(E\times F)-\dim^A(z)\,.\]
			Thus, inequality~\eqref{eq:hausint1} holds whenever $\dim^A(z)=n$. In particular, it holds when $z$ is Martin-L\"of random relative to $A$, i.e., for Lebesgue almost all $z\in\R^n$~\cite{LiVit08,Mart66}.
		\end{proof}

		The corresponding intersection formula for packing dimension has been shown for arbitrary $E,F\subseteq\R^n$ by Falconer~\cite{Falc94}. That proof is not difficult or long, but an algorithmic dimensional proof is presented here as an instance where this technique applies symmetrically to both Hausdorff and packing dimension.
		\begin{thm}[\cite{Falc94}]\label{thm:packint}
			For all $E,F\subseteq\R^n$, and for (Lebesgue) almost all $z\in\R^n$,
			\[\dimP(E\cap(F+z))\leq\max\{0,\dimP(E\times F)-n\}\,.\]
		\end{thm}
		\begin{proof}
			As in Theorem~\ref{thm:hausint}, we may assume that the intersection is nonempty. Apply Corollary~\ref{cor:p2suses}(3) to find a packing oracle $B\subseteq\N$ for $E\times F$, i.e.,
			\begin{equation}\label{eq:packintersect1}
				\dimP(E\times F)=\sup_{(x,y)\in E\times F}\Dim^{B}(x,y)
			\end{equation}
			and, given $\ve>0$, apply Corollary~\ref{cor:p2suses}(4) to $E\cap (F+z)$ to find a point $y\in E\cap(F+z)$ satisfying
			\begin{equation}\label{eq:packintersect2}
				\Dim^{B,z}(y)\geq\dimP(E\cap(F+z))-\ve\,.
			\end{equation}
			Then $(y,y-z)\in E\times F$, and we may proceed much as before:
			\begin{align*}
			\dimP(E\times F)&\geq\Dim^B(y,y-z)\tag{by equation~\eqref{eq:packintersect1}}\\
			&=\Dim^B(y,z)\tag{by Lemma~\ref{lem:pres}}\\
			&\geq \dim^B(z)+\Dim^{B}(y\mid z)\tag{by Theorem~\ref{thm:chainrule} relative to $B$}\\
			&\geq \dim^B(z)+\Dim^{B,z}(y)\tag{by Lemma~\ref{lem:relcond} relative to $B$}\\
			&\geq \dim^B(z)+\dimP(E\cap(F+z))-\ve\tag{by inequality~\eqref{eq:packintersect2}}\,.
			\end{align*}
			Again, $\dim^B(z)=n$ for almost every $z\in\R^n$, so this completes the proof.
		\end{proof}
	\section{Product Inequalities for Fractals}\label{sec:prod}

		In this section we give new proofs of four known product inequalities for fractal dimensions.	Inequality~\eqref{eq:prod1} below, which was stated in the introduction as Theorem~\ref{thm:mars}, is due to Marstrand~\cite{Mars54}. The other three inequalities are due to Tricot~\cite{Tric82}. Reference~\cite{Matt95} gives a more detailed account of this history.

		\begin{thm}	[\cite{Mars54,Tric82}] \label{thm:product}
			For all $E\subseteq\R^m$ and $F\subseteq\R^n$,
			\begin{align}
				\dimH(E)+\dimH(F)&\leq\dimH(E\times F)\label{eq:prod1}\\
				&\leq\dimP(E)+\dimH(F)\label{eq:prod2}\\
				&\leq\dimP(E\times F)\label{eq:prod3}\\
				&\leq\dimP(E)+\dimP(F)\label{eq:prod4}\,.
			\end{align}
		\end{thm}
		When $E$ and $F$ are Borel sets, it is simple to prove inequality~\eqref{eq:prod1} by using Frostman's Lemma, but the argument for general sets using net measures is considerably more difficult~\cite{Matt95,Falc85}. The new proof of inequality~\eqref{eq:prod1} given here applies to general sets but is similar in length to the previous proof for Borel sets.

		Although previous proofs of inequalities~\eqref{eq:prod1}--\eqref{eq:prod3} for general sets were already short, we present new, algorithmic information-theoretic proofs of those inequalities. Notice the superficial resemblance of Theorem~\ref{thm:product} to Theorem~\ref{thm:chainrule}. This similarity is not a coincidence; each inequality in Theorem~\ref{thm:product} follows from the corresponding line in Theorem~\ref{thm:chainrule}.
		\begin{proof}
			Corollary~\ref{cor:p2suses}(1) gives a Hausdorff oracle for $E\times F$, i.e., a set $A\subseteq\N$ such that
			\begin{equation}\label{eq:prodproof1}
				\dimH(E\times F) = \sup_{z\in E\times F}\dim^A(z)\,,
			\end{equation}
			and, for every $\ve>0$, Corollary~\ref{cor:p2suses}(2) gives points $x\in E$ and $y\in F$ such that
			\begin{equation}\label{eq:prodproof2}
				\dim^A(x)\geq\dimH(E)-\ve
			\end{equation}
			and
			\begin{equation}\label{eq:prodproof3}
				\dim^{A,x}(y)\geq\dimH(F)-\ve\,.
			\end{equation}

			Then we have
			\begin{align*}
				\dimH(E\times F)&\geq\dim^A(x,y)\tag{by equation~\eqref{eq:prodproof1}}
				\\&\geq\dim^A(x)+\dim^{A}(y\mid x)\tag{by Theorem~\ref{thm:chainrule} relative to $A$}
				\\&\geq \dim^A(x)+ \dim^{A,x}(y)\tag{by Lemma~\ref{lem:relcond} relative to $A$}
				\\&\geq\dimH(E)+\dimH(F)-2\ve\,.\tag{by inequalities~\eqref{eq:prodproof2} and~\eqref{eq:prodproof3}}
			\end{align*}
			Since $\ve>0$ was arbitrary, we conclude that inequality~\eqref{eq:prod1} holds.
			
			For inequality~\eqref{eq:prod2}, we use a packing oracle for $E$ and a Hausdorff oracle for $F$. Apply Corollary~\ref{cor:p2suses}(3) to 
			find a set $B\subseteq\N$ such that
			\begin{equation}\label{eq:prodproof4}
				\dimP(E) = \sup_{z\in E}\Dim^B(z)\,,
			\end{equation}
			and apply Corollary~\ref{cor:p2suses}(1) to find a set $C\subseteq\N$ such that
			\begin{equation}\label{eq:prodproof5}
				\dimH(F) = \sup_{z\in F}\dim^C(z)\,.
			\end{equation}
			Given $\varepsilon>0$, apply Corollary~\ref{cor:p2suses}(2) to find a point $(u,v)\in E\times F$ such that
			\begin{equation}\label{eq:prodproof6}
				\dim^{B,C}(u,v)\geq\dimH(E\times F)-\ve\,.
			\end{equation}
			Then we have
			\begin{align*}
				\dimP(E)+\dimH(F)&\geq\Dim^{B}(u)+\dim^{C}(v) \tag{by inequalities~\eqref{eq:prodproof4} and~\eqref{eq:prodproof5}}
				\\&\geq\Dim^{B,C}(u)+\dim^{B,C}(v)\tag{relativization cannot increase dimension}
				\\&\geq\Dim^{B,C}(u\mid v)+\dim^{B,C}(v)\tag{conditioning cannot increase dimension}
				\\&\geq\dim^{B,C}(u,v)\tag{by Theorem~\ref{thm:chainrule} relative to $B,C$}
				\\&\geq\dimH(E\times F)-\ve\,.\tag{by inequality~\eqref{eq:prodproof6}}
			\end{align*}
			Again, $\ve>0$ was arbitrary, so inequality~\eqref{eq:prod2} holds.
			
			For inequalities~\eqref{eq:prod3} and~\eqref{eq:prod4}, we use essentially the same arguments as above. We begin by finding a Hausdorff oracle for $E$ and packing oracles for $E\times F$ and $F$. Apply Corollary~\ref{cor:p2suses}(1) to find a set $B'\subseteq\N$ such that
			\begin{equation}\label{eq:prodproof7}
				\dimH(E) = \sup_{z\in E}\dim^{B'}(z)\,,
			\end{equation}
			and apply Corollary~\ref{cor:p2suses}(3) to find sets $A',C'\subseteq \N$ such that
			\begin{equation}\label{eq:prodproof8}
				\dimP(E\times F) = \sup_{z\in E\times F}\Dim^{A'}(z)\,,
			\end{equation}
			and
			\begin{equation}\label{eq:prodproof9}
				\dimP(F) = \sup_{z\in F}\Dim^{C'}(z)\,,
			\end{equation}

			Given $\varepsilon>0$, apply Corollary~\ref{cor:p2suses}(2) to find a point $y'\in F$ such that
			\begin{equation}\label{eq:prodproof10}
				\dim^{A'}(y')\geq\dimH(F)-\ve\,,
			\end{equation}
			and apply Corollary~\ref{cor:p2suses}(4) to find points $x',u'\in E$ and $v'\in F$ such that
			\begin{equation}\label{eq:prodproof11}
				\Dim^{A',y'}(x')\geq\dimP(E)-\ve
			\end{equation}
			and
			\begin{equation}\label{eq:prodproof12}
				\Dim^{B,C'}(u',v')\geq\dimP(E\times F)-\ve\,,
			\end{equation}
			where $B$ is as in equation~\eqref{eq:prodproof4}. Then we have
			\begin{align*}
				\dimP(E)+\dimP(F)&\geq\Dim^{B}(u')+\Dim^{C'}(v')\tag{by equations~\eqref{eq:prodproof4} and~\eqref{eq:prodproof9}}
				\\&\geq\Dim^{B,C'}(u')+\Dim^{B,C'}(v')\tag{relativization cannot increase dimension}
				\\&\geq\Dim^{B,C'}(u'\mid v')+\Dim^{B,C'}(v')\tag{conditioning cannot increase dimension}
				\\&\geq\Dim^{B,C'}(u',v')\tag{by Theorem~\ref{thm:chainrule} relative to $B,C'$}
				\\&\geq\dimP(E\times F)-\ve\tag{by inequality~\eqref{eq:prodproof12}}
				\\&\geq\Dim^{A'}(x',y')-\ve\tag{by equation~\eqref{eq:prodproof8}}
				\\&\geq\Dim^{A'}(x'\mid y')+\dim^{A'}(y')-\ve\tag{by Theorem~\ref{thm:chainrule} relative to $A'$}
				\\&\geq\Dim^{A',y'}(x')+\dim^{A'}(y')-\ve\tag{by Lemma~\ref{lem:relcond} relative to $A'$}
				\\&\geq\dimP(E)+\dimH(F)-3\ve\,.\tag{by inequalities~\eqref{eq:prodproof10} and~\eqref{eq:prodproof11}}
			\end{align*}
			Letting $\ve\to0$ completes the proof.
		\end{proof}

		\section{Product Characterization of Packing Dimension}\label{sec:packprod}

		Theorem~\ref{thm:product} is clearly tight in the sense that, for each inequality~\eqref{eq:prod1}--\eqref{eq:prod4}, there exist sets $E$ and $F$ that will make equality hold; for example, let $E=F=\emptyset$. In fact, for inequalities~\eqref{eq:prod1},~\eqref{eq:prod3}, and~\eqref{eq:prod4}, an even stronger tightness property holds: For every set $E\subseteq\R^m$, there is a set $F$ (e.g., any singleton) such that equality holds. Does this latter property hold for inequality~\eqref{eq:prod2}? The answer is much less obvious, but Bishop and Peres~\cite{BisPer96} proved that for every analytic set $E$, there are compact sets $F$ that make inequality~\eqref{eq:prod2} approach equality, in the sense that
		\[\dimP(E)=\sup\big(\dimH(E\times F)-\dimH(F)\big)\,,\]
		where the supremum is over all sets compact sets $F\subseteq\R^n$. We now prove that for arbitrary $E\subseteq\R^n$, there are (not necessarily compact) sets $F\subseteq\R^n$ that make inequality~\eqref{eq:prod2} approach equality.
		{\renewcommand{\thethm}{\ref{thm:packprod}}
		\begin{thm}
			For every set $E\subseteq\R^n$,
			\[\dimP(E)=\sup_{F\subseteq\R^n}\big(\dimH(E\times F)-\dimH(F)\big)\,.\]
		\end{thm}
		\addtocounter{thm}{-1}}
		\begin{proof}
		Let $E\subseteq\R^n$, and let $\varepsilon>0$. By inequality~\eqref{eq:prod2}, it suffices to show that there is some set $F\subseteq\R^n$ satisfying
		\begin{equation}\label{eq:thm:packprod:goal}
			\dimP(E)\leq \dimH(E\times F)-\dimH(F)+\varepsilon\,.
		\end{equation}

		We define the set $F\subseteq\R^n$ by
		\begin{equation}\label{eq:choiceofF}
			F=\{y\in\R^n:\dim(y)\leq n-\dimP(E)+\varepsilon\}\,.
		\end{equation}
		Then by Theorem~\ref{thm:p2s},
		\[\dimH(F)\leq \sup_{y\in F}\dim(y)\leq n-\dimP(E)+\varepsilon\,,\]
		so it only remains to show that
		\begin{equation}\label{eq:goal}
			\dimH(E\times F)\geq n\,.
		\end{equation}

		Corollary~\ref{cor:p2suses}(1) tells us that there is a Hausdorff oracle for $E\times F$, i.e., an oracle $A\subseteq\N$ such that
		\begin{equation}\label{eq:choiceofA}
			\dimH(E\times F)=\sup_{(x,y)\in E\times F}\dim^A(x,y)\,.
		\end{equation}
		We now construct a point $(x,y)\in E\times F$ whose dimension relative to $A$ is close to $n$. Let \begin{equation}\label{eq:choiceofdelta}
			\delta\in\left(0,\frac{\varepsilon}{2n+1}\right)\,.
		\end{equation}
		By Corollary~\ref{cor:p2suses}(4), there is some point $x\in E$ such that
		\begin{equation}\label{eq:choiceofx}
			\Dim^{A}(x)>\dimP(E)-\delta\,.
		\end{equation}

		We want the point $y\in F$ to complement $x$, in the sense that $y$'s information density should be high at precisions where $x$'s information density is low, so that the pair $(x,y)$ can never have information density much lower than $n$. At the same time, in order for $y$ to belong to $F$, $y$ needs to have low information density --- arbitrarily close to $n-\dimP(E)+\varepsilon$ --- at infinitely many precisions.

		To this end, we construct $y$ by appending zeros to its binary expansion when the information density of $(x,y)$ exceeds $n$, and appending random bits when that information density drops below $n$. This decision --- whether to append random bits or zeros --- is made at each of an exponentially increasing sequence of precisions.

		Formally, for each $j\in\N$, let $r_j=\lfloor\left(1+\delta\right)^j\rfloor$. Notice that $r_{j+1}-r_j=\delta r_j+O(1)$ and $\log(r_j)=\Theta(j)$. Define the point $y\in[0,1]^n$ by, for all $j\in\N$,
		\begin{equation*}\label{eq:y}
			\dot{y}[nr_j:nr_{j+1}]=
			\begin{cases}
				u_j
				&\text{if }K_{r_j}^{A}(x,y)> nr_j\\
				v_j
				&\text{otherwise}\,,
			\end{cases}
		\end{equation*}
		where $u_j=0^{nr_{j+1}-nr_j}$ and $v_j\in\{0,1\}^*$ is an $(nr_{j+1}-nr_j)$-bit string that is random relative to $A$, $x$, and $\dot{y}[0:nr_j]$, in the sense that $K^{A,x}(v_j\mid \dot{y}[0:nr_j])\geq nr_{j+1}-nr_j$. Hence, for all $r\leq r_{j+1}$,
		\begin{equation}\label{eq:vrandom}
			K^{A,x}(v_j[0:nr-nr_j]\mid \dot{y}[0:nr_j])= nr-nr_j+O(j)\,.
		\end{equation}

		Constructed in this way, the point $y$ has two key properties. First, the dimension and strong dimension of $(x,y)$ are both close to $n$. Second, access to information about $A$ and $x$ cannot affect the complexity of $y$ very much.  These properties are formalized in the following two lemmas.

	\begin{lem}\label{lem:xy}
	 	$n-\delta n \leq \dim^{A}(x,y)\leq\Dim^{A}(x,y)\leq n+2\delta n$.
	\end{lem}
	\begin{proof}
		First, we will prove by induction on $j$ that there is a constant $c$ such that, for all $j\geq 1$ and all integer precisions $r\in[r_{j-1},r_{j}]$,
		\begin{equation}\label{eq:nrj}
			(n-\delta n)r-cj^2\leq K^{A}_r(x,y)\leq (n+2\delta n) r+cj^2\,.
		\end{equation}
	
		Recall that $r_j=\lfloor (1+\delta)^j\rfloor$. We can choose $c$ large enough to make inequality~\eqref{eq:nrj} hold for $r=r_0=1$.
	
		Fix $j\geq 1$ and $r\in[r_j, r_{j+1}]$, and let $c$ be a sufficiently large constant satisfying inequality~\eqref{eq:nrj} for $r=r_j$. We consider two cases.

		First, if $K_{r_j}^{A}(x,y)\leq nr_j$,
		then $\dot{y}[nr_j:nr_{j+1}]=v_j$. Therefore, by equations~\eqref{eq:symmetry} and~\eqref{eq:convert}, assuming that $c$ is large enough to accommodate the implicit constants in each appearance of $O(j)$, we have
		\begin{align*}
			K_r^{A}(x,y)&\geq K_{r_j}^{A}(x,y)+K^{A}(v_j[0:nr-nr_j)]\mid \dot{(x,y)}[0:nr_j])+O(j)\\
			&\geq K_{r_j}^{A}(x,y)+nr-nr_j-O(j) \tag{by equation~\eqref{eq:vrandom}}\\
			&\geq (n-\delta n) r_j-cj^2+nr-nr_j-O(j) \tag{by inductive hypothesis}\\
			&=nr-\delta n r_j-cj^2-O(j)\\
			&\geq (n- \delta n) r - cj^2-O(j)\\
			&\geq (n-\delta n) r - c\cdot (j+1)^2
		\end{align*}
		and
		\begin{align*}
			K_r^{A}(x,y)&\leq K_{r_j}^{A}(x,y)+K^{A}(\dot{(x,y)}[2nr_j:2nr])+O(j)\\
			&\leq nr_j+K^{A}(\dot{(x,y)}[2nr_j:2nr])+O(j)\\
			&\leq nr_j+2nr-2nr_j+O(j)\\
			&\leq nr_j+2n(1+\delta)r_j-2nr_j+O(j)\\
			&= (n-2\delta n)r_j+O(j)\\
			&\leq (n-2\delta n)r+c\cdot (j+1)^2\,.
		\end{align*}

		Otherwise, if $K_{r_j}^{A}(x,y)> nr_j$, then $\dot{y}[nr_j:nr_{j+1}]=u_j$. Applying equations~\eqref{eq:symmetry} and~\eqref{eq:convert} and again assuming that $c$ is sufficiently large to accommodate the constants in each $O(j)$, we have
		\begin{align*}
			K_r^{A}(x,y)&\geq K^{A}_{r_j}(x,y)-O(j)\\
			&\geq nr_j-O(j)\\
			&\geq (1-\delta^2)n r_j-O(j)\\
			&= (n-\delta n)(1+\delta)r_j-O(j)\\
			&\geq (n-\delta n)r-O(j)\\
			&\geq (n-\delta n)r - c\cdot (j+1)^2
		\end{align*}
		and
		\begin{align*}
			K_r^{A}(x,y)&\leq K^{A}_{r_j}(x,y)+K^{A}(0^{nr-nr_j})+K^{A}(\dot{x}[nr:nr_j])+O(j)\\
			&\leq K^{A}_{r_j}(x,y)+nr-nr_j+O(j)\\
			&\leq nr_j+cj^2+nr-nr_j+O(j)\tag{by inductive hypothesis}\\
			&\leq nr+c\cdot (j+1)^2\,.
		\end{align*}
	
		Thus, by induction, there is a constant $c$ such that~\eqref{eq:nrj} holds for all $j\geq 1$ and all $r\leq r_j$. Since $j=O(\log r)$, dividing by $r$ and taking limits completes the proof of the lemma.
	\end{proof}

	\begin{lem}\label{lem:y}
		For all $r\in\N$, $K_r(y)\leq K_r^A(y\mid x)+O(\log r)$.
	\end{lem}
	\begin{proof}
		First observe that for all $j\in\N$
		\begin{equation}\label{eq:abxnohelp}
			K(\dot{y}[nr_j:nr_{j+1}]\mid \dot{y}[0:nr_j])\leq K^{A,x}(\dot{y}[nr_j:nr_{j+1}]\mid \dot{y}[0:nr_j])+O(j)\,.
		\end{equation}
		This is true because $\dot{y}[nr_j:nr_{j+1}]\in\{u_j,v_j\}$,
		\[K(u_j)\leq \log |u_j|+2\log\log|u_j|+O(1)=O(j)\,,\]
		and
		\[K(v_j)\leq |v_j|+2\log|v_j|+O(1)=|v_j|+O(j)\,,\]
		which is at most $O(j)$ greater than $K^{A,x}(v_j\mid \dot{y}[0:nr_j])$ by inequality~\eqref{eq:vrandom}.
		
		Now let $r\in\N$, and let $j\in\N$ satisfy $r_j\leq r<r_{j+1}$. For exactly the same reason as inequality~\eqref{eq:abxnohelp},
		\begin{equation}\label{eq:abxnohelplast}
			K(\dot{y}[nr_j:nr]\mid \dot{y}[0:nr_j])\leq K^{A,x}(\dot{y}[nr_j:nr]\mid \dot{y}[0:nr_j])+O(j)\,.
		\end{equation}
		We now use inequalities~\eqref{eq:abxnohelp} and~\eqref{eq:abxnohelplast} together with equation~\eqref{eq:convert} and a chained application of the symmetry of information:
		\begin{align*}
		K_r(y)&=K(\dot{y}[0:nr])+O(j)\\
		&= K(\dot{y}[nr_j:nr]\mid \dot{y}[0:nr_j])+\sum_{i=0}^{j-1} \big(K(\dot{y}[nr_i:nr_{i+1}]\mid \dot{y}[0:nr_i])+O(i)\big)+O(j)\\
		&\leq K^{A,x}(\dot{y}[nr_j:nr]\mid \dot{y}[0:nr_j])+\sum_{i=0}^{j-1} \big(K^{A,x}(\dot{y}[nr_i:nr_{i+1}]\mid \dot{y}[0:nr_i])+O(i)\big)+O(j)\\
		&=K^{A,x}(\dot{y}[0:nr])+O(j^2)\\
		&=K^{A,x}_r(y)+O(\log r)\\
		&\leq K^A_r(y\mid x)+O(\log r)\,,
	\end{align*}
	where the last inequality is an application of Lemma~\ref{lem:relcond} relative to $A$.
	\end{proof}

	We will use two simple consequences of Lemma~\ref{lem:y} in terms of dimension, as given in the following corollaries.

	\begin{cor}~\label{cor:dimy}
		$\dim(y)\leq \dim^A(y)$.
	\end{cor}
	\begin{proof}
		Dividing the conclusion of Lemma~\ref{lem:y} by $r$ and taking limits inferior gives $\dim(y)\leq\dim^{A,x}(y)$. Since relativization cannot increase dimension, $\dim^{A,x}(y)\leq\dim^A(y)$. 
	\end{proof}
	\begin{cor}\label{cor:dimax}
		$\Dim^A(x)\leq\Dim^A(x\mid y)$
	\end{cor}
	\begin{proof}
		By symmetry of information,
		\[K_r^A(x)+K_r^A(y\mid x)\leq K^A(y)+K^A(x\mid y)+O(\log r)\,.\]
		Combining this with Lemma~\ref{lem:y} gives
		\[K_r^A(x)\leq K_r^A(x\mid y)+ O(\log r)\,.\]
		Dividing by $r$ and taking limits superior shows that $\Dim^A(x)\leq \Dim^A(x\mid y)$.
	\end{proof}

	We are now ready to complete the proof of Theorem~\ref{thm:packprod}. First,
	\begin{align*}
		\dim(y)&\leq\dim^{A}(y)\tag{by Corollary~\ref{cor:dimy}}\\
		&\leq\Dim^{A}(x,y)-\Dim^{A}(x\mid y)\tag{by Theorem~\ref{thm:chainrule} relative to $A$}\\
		&=\Dim^{A}(x,y)-\Dim^{A}(x)\tag{by Corollary~\ref{cor:dimax}}\\
		&\leq n+2\delta n-\Dim^{A}(x)\tag{by Lemma~\ref{lem:xy}}\\
		&< n+2\delta n-(\dimP(E)-\delta)\tag{by inequality~\eqref{eq:choiceofx}}\\
		&< n-\dimP(E)+\varepsilon\,,\tag{by inequality~\eqref{eq:choiceofdelta}}
	\end{align*}
	so $y\in F$. Hence,
	\begin{align*}
		\dimH(E\times F)&\geq \dim^{A}(x,y)\tag{by equation~\eqref{eq:choiceofA}}\\
		&\geq n-3\delta n\tag{by Lemma~\ref{lem:xy}}\,.
	\end{align*}
	Letting $\delta\to 0$, we see that inequality~\eqref{eq:goal} holds.
\end{proof}
		
	\section{Conclusion}
		The applications of theoretical computer science to pure mathematics in this paper yielded significant generalizations to basic theorems on Hausdorff and packing dimensions, as well as simpler arguments for other such theorems. Understanding classical fractal dimensions as pointwise, algorithmic information theoretic quantities enables reasoning about them in a way that is both fine-grained and intuitive, and the proofs in this work are further evidence of the power and versatility of bounding techniques using the point-to-set principle (Theorem~\ref{thm:p2s}). In particular, Theorems~\ref{thm:hausint} and~\ref{thm:packprod} demonstrate that this approach can be used to strengthen the foundations of fractal geometry. Therefore, in addition to further applications of these techniques, broadening and refining our understanding of the relationship between classical geometric measure theory and Kolmogorov complexity is an appealing direction for future investigations.
\bibliography{fipad}

\end{document}